\documentclass[aps,prl,twocolumn,superscriptaddress,floatfix,nofootinbib,showpacs,longbibliography,groupedaddress]{revtex4-1}

\usepackage[utf8]{inputenc}  
\usepackage[T1]{fontenc}     %Output what you want e.g., é, ł, a, ü
\usepackage[british]{babel}  %Do hyphenation according to british english
\usepackage[sc,osf]{mathpazo}
\usepackage[scaled=0.86]{berasans}  % URL font that go well wtih palatino
\usepackage[colorlinks=true, citecolor=blue, urlcolor=blue]{hyperref}  
\usepackage{graphicx} % Package to insert exteral figures
\usepackage[babel]{microtype}  %Improves text justification
\usepackage{amsmath,amssymb,amsthm,bm,amsfonts,mathrsfs,bbm} %Usefull math packages

\usepackage{xspace}  %Useful to add space in macros
\usepackage{pgf,tikz}
\usepackage{xcolor}
\usepackage{multirow}
\usepackage{array}
\usepackage{bigstrut}
\usepackage{braket}
\usepackage{color}
\usepackage{natbib}
\usepackage{multirow}
\usepackage{mathtools}
\usepackage{float}
\usepackage[caption = false]{subfig}
\usepackage{xcolor,colortbl}
\usepackage{color}

\newcommand{\be}{\begin{equation}}
\newcommand{\ee}{\end{equation}}
\newcommand{\ba}{\begin{eqnarray}}
\newcommand{\ea}{\end{eqnarray}}

\newtheorem{theorem}{Theorem}
\newtheorem{corollary}{Corollary}
\newtheorem{definition}{Definition}
\newtheorem{proposition}{Proposition}

\def\>{\rangle}
\def\<{\langle}

\usepackage{centernot}
\usepackage{subfig}

\begin{document}

\title{Multi-copy adaptive local discrimination: Strongest possible two-qubit nonlocal bases}

\author{Manik Banik}
%\email{manik11ju@gmail.com}    
\affiliation{School of Physics, IISER Thiruvananthapuram, Vithura, Kerala 695551, India.}

\author{Tamal Guha}
%\email{xyz}    
\affiliation{Physics and Applied Mathematics Unit, Indian Statistical Institute, 203 B.T. Road, Kolkata 700108, India.}

\author{Mir Alimuddin}
%\email{xyz}    
\affiliation{Physics and Applied Mathematics Unit, Indian Statistical Institute, 203 B.T. Road, Kolkata 700108, India.}

\author{Guruprasad Kar}
%\email{xyz}    
\affiliation{Physics and Applied Mathematics Unit, Indian Statistical Institute, 203 B.T. Road, Kolkata 700108, India.}

\author{Saronath Halder}
%\email{saronath.halder@gmail.com}    
\affiliation{Harish-Chandra Research Institute, HBNI, Chhatnag Road, Jhunsi, Allahabad 211019, India.} 

\author{Some Sankar Bhattacharya}
%\email{xyz} 
\affiliation{Department of Computer Science, The University of Hong Kong, Pokfulam Road, Hong Kong.}

\begin{abstract}
Ensembles of composite quantum states can exhibit nonlocal behaviour in the sense that their optimal discrimination may require global operations. Such an ensemble containing $N$ pairwise orthogonal pure states, however, can always be perfectly distinguished under adaptive local scheme if $(N-1)$ copies of the state are available. In this letter, we provide examples of orthonormal bases in two-qubit Hilbert space whose adaptive discrimination require $3$ copies of the state. For this composite system we analyze multi-copy adaptive local distinguishability of orthogonal ensembles in full generality which in turn assigns varying nonlocal strength to different such ensembles. We also come up with ensembles whose discrimination under adaptive separable scheme require less number of copies than adaptive local schemes. Our construction finds important application in multipartite secret sharing tasks and indicates towards an intriguing super-additivity phenomenon for locally accessible information. 
\end{abstract}

%\pacs{03.65.Ta,03.65.Ud, 03.67.Dd, 03.67.Hk}
%\keywords{\textclor{red}{Need to be added}}

% 03.65.Ta	Foundations of quantum mechanics;
% 03.67.Dd	Quantum cryptography and communication security
% 03.67.Hk	Quantum communication
% 03.65.Ud	Entanglement and quantum nonlocality
\maketitle	
{\it Introduction.--} Second quantum revolution aims to harness individual quantum system for storing, transferring, and manipulating information \cite{Chuang}. In many of the information protocols the elementary step is reliable decoding of the classical information encoded in some physical system. When involved systems are quantum, several interesting observations appear that otherwise are not present in classical world. For instance, classical message encoded in the states of a composite quantum system may not be completely retrieved under local quantum operations and classical communications (LOCC) among the spatially separated subsystems. Such a set of states is called nonlocal as their optimal discrimination require global operation(s) on the composite system.  In a seminal paper \cite{Bennett99}, Bennett {\it et al.} provided examples of orthonormal bases in $(\mathbb{C}^3)^{\otimes2}$ and $(\mathbb{C}^2)^{\otimes3}$ that are not perfectly distinguishable under LOCC. These examples are quite striking as they contain only product states and hence introduced the phenomenon of `quantum nonlocality without entanglement'. For the simplest multipartite system $(\mathbb{C}^2)^{\otimes2}$, example of LOCC indistinguishable ensemble was first identified in \cite{Ghosh01}. Unlike Bennett {\it et al.}'s examples, an orthogonal ensemble in $(\mathbb{C}^2)^{\otimes2}$ must contain entangled state(s) for local indistinguishability \cite{Walgate02}. Historically, a difference between global and local distinguishability for a two-qubit ensemble containing only product states was conjectured by Peres \& Wootters \cite{Peres91}. Importantly, Peres-Wootters' ensemble contains non-orthogonal states and recently their conjecture has been proven to be true \cite{Chitambar13}. The results in \cite{Bennett99,Ghosh01,Walgate02,Peres91} initiated a plethora of studies on the local state discrimination problem that turn out to be deeply interlinked with quantum entanglement theory (see \cite{Bennett99(1),Walgate00,Horodecki03,DiVincenzo03,Ghosh04,Watrous05,Niset06,Duan07,Calsamiglia10,Bandyopadhyay11,Chitambar14,Halder18,Demianowicz18,Halder19,Agrawal19,Rout19,Rout20,Horodecki09} and references there in). Subsequently, local indistinguishability has been identified as crucial primitive for cryptography protocols, such as quantum data hiding \cite{Terhal01,Eggeling02,Matthews09} and quantum secret sharing \cite{Hillery99,Markham08,Rahaman15}; it also underlies the non-zero gap between single-shot and multi-shot classical capacities of noisy quantum channels \cite{Fuchs}, and has been shown to have foundational appeal \cite{Pusey12,Bandyopadhyay17,Bhattacharya20}. Mathematical difficulty of characterizing the class of LOCC operations \cite{Chitambar14(1)} motivates researchers to address the state discrimination problem with a larger class of operations, namely separable super-operator (SEP) and/or the set of operations that map positive partial transposition (PPT) states to PPT states \cite{Duan09,Yu12,Yu14,Bandyopadhyay15,Cohen15}. 

The problem of local discrimination discussed so far considers identifying the unknown state from a single copy of the system. Given many copies of the state $\ket{\psi_i}^{\otimes k}$ the probability of knowing the state $\ket{\psi_i}$ and consequently guessing the classical message `$i$' increases; $\ket{\psi}_i$ is sampled from some known ensemble $\mathcal{E}_N\equiv\{p_i,\ket{\psi_i}~|~p_i>0~\&~\sum p_i=1\}_{i=1}^N$. If there is no limitation on the number $k$, then the state can be identified exactly even without any prior knowledge from which ensemble it is sampled. This fact, known as local tomography, has been axiomatized by many authors in physical reconstruction of Hilbert space quantum theory \cite{Chiribella11,Barnum14,Niestegge20}. However, given the prior knowledge of the orthogonal pure states ensemble $\mathcal{E}_N$, only $(N-1)$ copies are sufficient for perfect local discrimination \cite{Walgate00,Bandyopadhyay11}. In this multi-copy scenario, the parties are allowed to invoke adaptive as well as non-adaptive strategies for discrimination.
\begin{definition}\label{def1}
(Local adaptive strategy) Given multiple copies of the state in an adaptive strategy each of the copies are addressed individually and the maximal possible information regarding the unknown state is extracted through LOCC. Depending on the knowledge obtained from the former copies strategies on the later copies are modified adaptively.
\end{definition}
Such adaptive strategies have already been studied extensively in quantum channel discrimination \cite{Hayashi09,Harrow10,Pirandola19} and in noise estimation \cite{Pirandola17,Cope17}. Perfect local discrimination of any orthogonal ensemble $\mathcal{E}_N$ is possible with $(N-1)$ copies under such adaptive schemes. The known examples of the nonlocal ensemble, however, do not even require $(N-1)$ copies for perfect discrimination. In fact, the authors in Ref. \cite{Bennett99} found no example where more than two copies of the unknown state are needed for perfect discrimination and thus they laid down the question: ``Are there any sets of states, entangled or not, for which some finite number (greater than $2$) of copies of the state is necessary for distinguishing the states reliably?" Even after more than two decades, to the best of our knowledge, there is no conclusive answer to this question. In the present letter we address this question and show that there indeed exist ensembles of orthogonal pure states for composite system that require more than two copies of the state for perfect local discrimination under adaptive strategies. Our explicit constructions form orthonormal bases for two-qubit Hilbert space. We, in-fact, completely characterize the two-qubit ensembles that require three copies of the state for perfect local discrimination under adaptive strategies. Interestingly, we also find ensembles that require three copies of the state for perfect discrimination under adaptive LOCC whereas two copies suffice if adaptive SEP protocols are followed. Muli-copy consideration of local state discrimination problem establishes the `more nonlocality with less entanglement' phenomenon for $(\mathbb{C}^2)^{\otimes2}$ ensembles which was earlier reported for ensembles in $(\mathbb{C}^3)^{\otimes2}$ \cite{Horodecki03}. Our study indicates an intriguing super-additivity phenomenon for the {\it locally accessible information} of multi-site quantum ensembles, and we discuss a practical application of the present constructions in multipartite secret sharing task.  

{\it Results.--} Throughout the paper we follow the standard notations used in quantum information community. For instance, $\{\ket{0},\ket{1}\}$ is the computational bases (eigenkets of Pauli $\sigma_z$ operator) of $\mathbb{C}^2$ and $\ket{ab}$ stands for the short hand notation of the bipartite state $\ket{a}_A\otimes\ket{b}_B\in\mathbb{C}^2_A\otimes\mathbb{C}^2_B$. Given single copy of the state, the authors in \cite{Walgate02} have analyzed LOCC distinguishability of orthogonal ensembles in $(\mathbb{C}^2)^{\otimes2}$. Here we will analyse their multi-copy local distinguishability. We will explore the case where Alice and Bob follow adaptive protocols for discrimination. An ensemble of three orthogonal pure states (of arbitrary composite system) is always two-copy distinguishable under adaptive LOCC \cite{Walgate00,Bandyopadhyay11}. We will therefore focus on the ensembles $\mathcal{E}_4=\{\ket{\psi_i}\}_{i=1}^4\subset(\mathbb{C}^2)^{\otimes2}$, where $\langle\psi_i|\psi_j\rangle=\delta_{ij}$. Throughout the paper we will consider that the states are uniformly sampled from the ensemble. Given multiple copies, while discriminating such a set under adaptive protocol, following three situations can arise: (a) the states can be discriminated after acting on the first copy at an early stage of the protocol, (b) protocol on the first copy results in a conclusion that the given state belongs to the group $\mathcal{G}_l\equiv\{\ket{\psi_l}\}$ or in its complement group $\mathcal{G}_l^{\mathsf{C}}:=\mathcal{E}_4\setminus\mathcal{G}_l$, for some $l\in\{1,\cdots,4\}$, and (c) protocol on the first copy results in a conclusion that the state belongs to the group $\mathcal{G}_{ij}\equiv\{\ket{\psi_i},\ket{\psi_j}\}$ for some $i,j\in\{1,\cdots,4\}$ with $i\neq j$. In the second and third cases the discrimination has to be completed from the subsequent copies. This leads us to our first result.
\begin{theorem}\label{theo1}
An orthogonal ensemble $\mathcal{E}_4\subset(\mathbb{C}^2)^{\otimes2}$ is two-copy distinguishable under adaptive LOCC protocol if and only if one of the following is satisfied:
\begin{itemize}
\item[(i)] the set is one-copy distinguishable and second copy is not at all required (trivial case); 
\item[(ii)] protocol on the first copy leads to a conclusion that the given state belongs to a group $\mathcal{G}_l$ or $\mathcal{G}_l^{\mathsf{C}}$ for some $l\in\{1,\cdots,4\}$ such that $\mathcal{G}_l^{\mathsf{C}}$ contains no more than one entangled state.
\item[(iii)] protocol on the first copy leads to a conclusion that the given state belongs to the a group $\mathcal{G}_{ij}$ for some $i,j\in\{1,\cdots,4\}$ with $i\neq j$, such that the projectors $\mathbb{P}_{ij}$ and $\mathbb{P}_{ij}^{\mathsf{C}}$ are separable, where $\mathbb{P}_{ij}:=\ket{\psi_i}\bra{\psi_i}+\ket{\psi_j}\bra{\psi_j}~\&~\mathbb{P}_{ij}^{\mathsf{C}}:=\mathbf{I}_4-\mathbb{P}_{ij}$.
\end{itemize}
\end{theorem}
\begin{proof}
For perfect discrimination, after the protocol on first copy either the state is discriminated straightaway [case (i)] or some states (at least one) must be eliminated [the other two cases]. In case (ii), if the protocol on the first copy leads to the conclusion that that the given set is in some group $\mathcal{G}_l$ then the discrimination is done. If it leads to a conclusion in the complement group $\mathcal{G}_l^{\mathsf{C}}$ (which will be the case with some non-zero probability) then the unknown state needs to be distinguished from the second copy. Recall that three orthogonal states in $(\mathbb{C}^2)^{\otimes2}$ can be exactly locally distinguished if and only if at least two of
those states are product states \cite{Walgate02}, and hence leading us to the assertion (ii). In case (iii), if the protocol on first copy leads to a conclusion that the unknown state belongs to either in some group $\mathcal{G}_{ij}$ or in its complement group $\mathcal{G}_{ij}^{\mathsf{C}}$ then protocol on the second copy perfectly succeed as two orthogonal states (in any dimension) can always be exactly locally distinguished \cite{Walgate00}. Suppose that there is some LOCC protocol on the first copy that leads us to this desired conclusion. Such a protocol will perfectly distinguish the density operators $\rho:=\frac{1}{2}\mathbb{P}_{ij}$ and $\rho^\perp:=\frac{1}{2}\mathbb{P}_{ij}^{\mathsf{C}}$. Recall that two rank-two orthogonal density operators $\rho,\rho^\perp\in\mathcal{D}(\mathbb{C}^2\otimes\mathbb{C}^2)$ are LOCC distinguishable if and only if the projectors onto the supports of each of the
mixed states are separable \cite{Chitambar14}, and hence leads us to the assertion (iii). 

To complete the argument, let us consider that after the protocol on first copy, none of the states is eliminated, rather the knowledge regarding the ensemble gets updated, {\it i.e.} the uniform ensemble  $\mathcal{E}_4\equiv\{1/4,\ket{\psi_i}\}$ ends up into some non-uniform ensemble $\mathcal{E}^\prime_4\equiv\{p_i,\ket{\psi_i}~|~p_i\neq0~\forall~i\}$. In such a case, the ensemble $\mathcal{E}^\prime$ must be discriminated from the second copy of the state. According to the Theorem 4 of Ref. \cite{Walgate02}, this can be done if and only if all the four states are product. This is because Theorem 4 of Ref. \cite{Walgate02} (and also the other theorems therein) is independent of the prior probability distributions. This completes the proof. 
\end{proof}
Given the above theorem, naturally the question arises which orthogonal ensembles of $(\mathbb{C}^2)^{\otimes2}$ are two-copy distinguishable under adaptive LOCC. As an immediate corollary of Theorem \ref{theo1} we first have the following result (Proof provided in the Supplemental).
\begin{corollary}
Any orthogonal ensemble of $(\mathbb{C}^2)^{\otimes2}$ containing no more than two entangled states is two-copy distinguishable under adaptive LOCC. 
\end{corollary}
Ensembles having more than two entangled states can not be in category (i) or category (ii) of Theorem \ref{theo1}. However, they can be in category (iii). For instance, consider the orthonormal basis $\mathcal{B}^{[\theta]}$ with basis states:
\begin{align*}
\ket{\phi^+_\theta}&:=S_\theta\ket{00}+C_\theta\ket{11},~\ket{\phi^-_\theta}:=C_\theta\ket{00}-S_\theta\ket{11},\\
\ket{\psi^+_\theta}&:=S_\theta\ket{01}+C_\theta\ket{10},~\ket{\psi^-_\theta}:=C_\theta\ket{01}-S_\theta\ket{10},
\end{align*}
where $S_\theta\equiv\sin\theta~\&~C_\theta\equiv\cos\theta$ with $0\le\theta\le\pi/2$. All the states are entangled whenever $\theta\neq0,\pi/2$ and in particular $\theta=\pi/4$ corresponds to the maximally entangled basis/ Bell basis. These entangled bases are two-copy distinguishable under adaptive LOCC. The protocol goes as follows: on the first copy both Alice and Bob perform $\sigma_z$ measurement on their part of the composite system and compare their measurement results. Correlated outcomes imply that the given state belongs to the group $\mathcal{G}_{\ket{\phi}}\equiv\{\ket{\phi^+_\theta},\ket{\phi^-_\theta}\}$ whereas anti-correlated outcomes imply that it is from the group $\mathcal{G}_{\ket{\psi}}\equiv\{\ket{\psi^+_\theta},\ket{\psi^-_\theta}\}$. The result of Walgate {\it et al.} \cite{Walgate00} assures prefect local distinguishability of the states from the second copy. Let us now construct another orthonormal basis 
\begin{align*}
\rotatebox[origin=c]{0}{$\mathcal{A}_\gamma^{[\alpha,\beta]}\equiv$}
\left\{\!\begin{aligned}
\ket{a_1}&:=\ket{\phi^-_\alpha},~~\ket{a_3}:=S_\gamma\ket{\phi^+_\alpha}+C_\gamma\ket{\psi^+_\beta}\\
\ket{a_2}&:=\ket{\psi^-_\beta},~~\ket{a_4}:=C_\gamma\ket{\phi^+_\alpha}-S_\gamma\ket{\psi^+_\beta}
\end{aligned}\right\},	
\end{align*}
with $\alpha,\beta,\gamma\in[0,\pi/2]$. Whenever $\alpha,\beta\neq0,\pi/2$, the states $\ket{a_1}~\&~\ket{a_2}$ are entangled. However entanglement of the other two states demand further restrictions on the parameters. For instance, the states $\ket{a_3},\ket{a_4}\in\mathcal{A}_{\pi/4}^{[\alpha,\beta]}$ are entangled {\it iff} $\alpha\neq\beta,~\pi/2-\beta$. In the rest of the paper we will mainly analyze the case $\gamma=\pi/4$. For generic consideration of $\gamma$ see the supplementary material \cite{Supple}. We also consider the cases where all the states in $\mathcal{A}_{\pi/4}^{[\alpha,\beta]}$ are entangled. Next we report an interesting feature of the set $\mathcal{A}_{\pi/4}^{[\alpha,\beta]}$.
\begin{proposition}\label{prop1}
$\mathbb{V}[ij]:=Span\{\ket{a_i},\ket{a_j}\}$ with $i\neq j$ and $\ket{a_i},\ket{a_j}\in\mathcal{A}_{\pi/4}^{[\alpha,\beta]}$; $\alpha\neq\beta,~\pi/2-\beta$. For every choice of $i,j\in\{1,2,3,4\}$ the projector onto the two dimensional subspace $\mathbb{V}[ij]$ is entangled. 
\end{proposition}
See supplementary material \cite{Supple} for the proof. Coming back to local distinguishability of the set $\mathcal{A}_{\pi/4}^{[\alpha,\beta]}$, it is neither in category (i) nor in category (ii) of Theorem \ref{theo1}. Furthermore, Proposition \ref{prop1} obstructs it to be in category (iii) too. These altogether guide us to the following result. 
\begin{theorem}\label{theo2}
The set $\mathcal{A}_{\pi/4}^{[\alpha,\beta]}$, with $\alpha\neq\beta,~\pi/2-\beta$, is perfectly distinguishable under adaptive LOCC protocol if and only if three copies of the state are available. 
\end{theorem}
Under the adaptive local discrimination scheme we therefore have a conclusive answer to the Bennett {\it et al.}' question \cite{Bennett99} discussed in introduction -- there indeed exists ensemble of states that requires more than two copies of the state for perfect discrimination under adaptive local scheme. Under such schemes the set $\mathcal{A}_{\pi/4}^{[\alpha,\beta]}$ can be considered as the strongest nonlocal ensemble (of orthogonal pure states) in $(\mathbb{C}^2)^{\otimes2}$. Naturally the question arises whether such a nonlocal ensemble requires all of its states to be entangled. Our next theorem (proof provided in supplementary material \cite{Supple}) answers this question. 
\begin{theorem}\label{theo3}
Any orthonormal basis of $(\mathbb{C}^2)^{\otimes2}$ containing exactly three entangled states requires three copies of the state for perfect discrimination under adaptive LOCC. 
\end{theorem}
Theorem \ref{theo3} is quite interesting in a sense. It tells that any basis containing three entangled states is two-copy indistinguishable under adaptive LOCC, whereas some bases with all entangled states (e.g. Bell basis) are two-copy distinguishable under local adaptive scheme. This mimics the phenomenon of {\it more nonlocality with less entanglement} \cite{Horodecki03} as a set with less number of entangled states turns out to be harder to discriminate. Importantly, the ensemble in Ref. \cite{Horodecki03} exhibiting this feature lives in $(\mathbb{C}^3)^{\otimes2}$ and with single-copy consideration such a phenomenon is not possible in $(\mathbb{C}^2)^{\otimes2}$. Multi-copy adaptive discrimination makes this phenomenon possible for two-qubit ensembles and in-fact introduces a hierarchical strength in the {\it more nonlocality with less entanglement} phenomenon.       

Next we consider multi-copy discrimination of nonlocal ensembles under the adaptive SEP. Recall that this set of operations strictly includes the set of LOCC operations \cite{Chitambar14(1)}, even for two-qubit Hilbert space \cite{Duan09}. Our next result, however, establishes that given two copies of the state even this larger class of operations fails to erase the nonlocal feature of some ensembles. 
\begin{theorem}\label{theo4}
There exist suitable choices of the parameters $\alpha,\beta$, and $\gamma$ such that under  adaptive SEP protocol perfect discrimination of the set $\mathcal{A}_{\gamma}^{[\alpha,\beta]}$ requires three copies of the state.  
\end{theorem}
For instance, $\alpha,\beta\neq0,\pi/2$ and $\gamma=\pi/6$ constitute such choices. Generic analysis is referred to the supplementary material \cite{Supple}. Similar to Theorem \ref{theo1}, any protocol on the first copy that just updates the prior probability distribution without eliminating state(s) does not work good in this case too. Theorem \ref{theo4} is also important from another perspective. Mathematical characterization of local operations is extremely hard in general. While implementation of such an operation may be done by some finite round of LOCC protocol, its implementation may also demand infinite round of protocol. One can also define topological closure of different such sets. However all these different classes of local operations are strictly included within the set of SEP protocol \cite{Chitambar14(1)}. Therefore the ensembles in Theorem \ref{theo4} remain two copy indistinguishable under adaptive protocol even under infinite round of LOCC. Arguably, these are the strongest two-qubit nonlocal ensembles under any adaptive discrimination scheme - LOCC and/or SEP. Considering general  $\mathcal{A}_{\gamma}^{[\alpha,\beta]}$ we, indeed, find ensembles that are perfectly two-copy distinguishable under adaptive SEP, whereas adaptive LOCC requires three copies for perfect discrimination (ensembles in Theorem \ref{theo2} are such examples).
\begin{figure}[t!]
\includegraphics[scale=0.5]{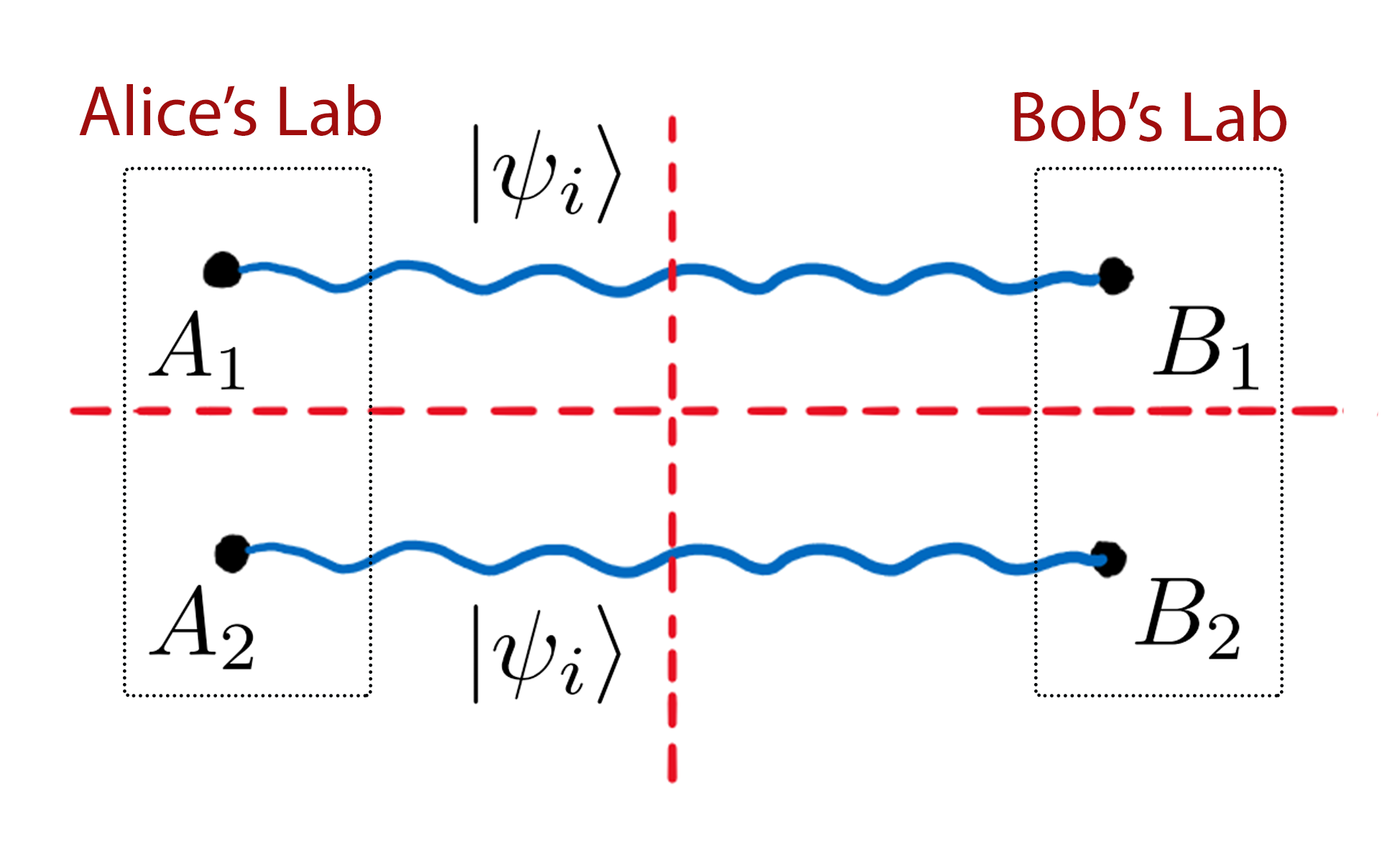}
\caption{(Color online) Alice and Bob share multiple copies of an unknown state $|\psi_i\rangle$ belonging to a known ensemble $\mathcal{E}_N$. Availability of additional entanglement resource across the vertical (party-cut) and/or horizontal (copy-cut) dotted lines give rise to different sets of allowed operations performed by Alice and Bob- a) no additional entanglement across the vertical and horizontal lines gives adaptive LOCC, b) adaptive SEP (more generally entangling operation) requires additional entanglement across vertical line but no entanglement across horizontal line, c) additional entanglement across horizontal line but no entanglement across vertical line gives non-adaptive LOCC, d) non-adaptive SEP requires entanglement across both vertical and horizontal lines.}\label{fig}
\end{figure}

We are therefore left with the only possibility whether the aforesaid ensembles are two-copy distinguishable under a non-adaptive local protocol. In such a protocol Alice (Bob) addresses both the systems at her (his) end simultaneously to perform some global measurement and communicates the outcomes with the other party. Such measurement includes SEP as well as entangled basis measurements. Important to note that, existence of an ensemble that is distinguishable under multi-copy local non-adaptive scheme but indistinguishable under adaptive local scheme will establish an intriguing super-additivity behaviour for the locally accessible information of multi-site quantum ensembles \cite{Badziag03}. Our constructions are therefore quite interesting -- either they will establish two-copy indistinguishability (under generic protocol) or it will demonstrate super-additivity of locally accessible information. In this regard, super-additivity phenomena of classical capacity of a noisy quantum channel is worth mentioning \cite{Fuchs97,King01,Schumacher01,King02}. There also entangled decoding on multiple copies of noisy channels turns out to be more efficient than its multiple single-copy use. 

An ensemble indistinguishable under local non-adaptive scheme will also be indistinguishable under adaptive LOCC. However, such an ensemble may be perfectly distinguished with adaptive SEP. Note that, physical realization of both the local non-adaptive scheme as well as adaptive SEP scheme require entanglement. While in the first case entanglement across the cut between different copies, in the later case entanglement is consumed across the cut between parties (see Fig.\ref{fig}). However, at present we have no intuition for construction of such ensembles and leave this question open for further research. In the following we rather focus on possible application of the present constructions.          

{\it Multiparty secret sharing.--} Secret sharing is an important cryptographic primitive \cite{Shamir79}. In $(k,m)$ secret sharing scheme an administrator wishes to distribute $k$-bit classical message among $m$ spatially separated parties in such a way that perfect revelation of the message requires (classical) collaboration among all the $m$ parties. In quantum scenario, administrator encodes $k$-bit massage in some $m$-partite state, {\it i.e.} $i\mapsto\rho_i\in\mathcal{D}(\mathcal{H}^{\otimes m})$ such that $\{\rho_i\}_{i=1}^{2^k}$ is a mutually orthogonal set of states \cite{Hillery99}. Our construction turns out to be efficient for a $(2,6)$ pure quantum protocol (encoded states are pure). The administrator encodes her $2$-bit message into an ensemble $\mathcal{E}_{4}\equiv\{\ket{\psi_i}\}_{i=1}^4$ which is indistinguishable under two-copy adaptive LOCC, prepares three copies of the state and distributes those among the six separated parties. In short, $i\mapsto\ket{\psi_i}\mapsto\ket{\psi_i}_{A_1B_1}\otimes\ket{\psi_i}_{A_2B_2}\otimes\ket{\psi_i}_{A_3B_3}$. No four parties can exactly decode the message through classical collaboration among them \cite{Self}.     

In a stronger variant of the aforesaid task even the classical collaborations among the parties are insufficient for perfect revelation of the message -- the parties need to come together to decode the information. It is also possible to come up with such a $(1,2)$ protocol from our construction. Recall that the result of Walgate {\it et al.} \cite{Walgate02} prohibits any such $(1,m)$ pure quantum protocol. So, the administrator choose some pair of index $i,j$ as in Proposition \ref{prop1} and then encodes her massage into the density operators $\sigma(\lambda):=\lambda\ket{a_i}\bra{a_i}+(1-\lambda)\ket{a_j}\bra{a_j}$ supported on the subspace $\mathbb{V}[ij]$ and its orthogonal pair $\sigma^\perp(\mu)$ supported on the complementary subspace. Whenever $\lambda,\mu\in(0,1)$, the perfect revelation of the massage demands both the parties to come together to apply a global discrimination measurement. At this point it may be interesting to find the values of $\lambda~\&~\mu$ that will ensure least amount of information accessible by LOCC. 

{\it Discussion.--} We have discoursed the problem of local state discrimination when multiple copies of the unknown state are available. Our study is a major advancement in this direction as it establishes that there exist ensembles of pure composite orthogonal states that are not two copy distinguishable under adaptive LOCC protocol. Under multi-copy adaptive scheme we have completely characterize the nonlocal behavior of the orthogonal ensembles of two-qubit Hilbert space. Such two-copy indistinguishable sets find useful applications in multi-party secret sharing tasks.    

Our study also raises a number of interesting questions. A conclusive answer to the possible super-additivity behavior for {\it locally accessible information} is worth deserving. A non-adaptive protocol is still missing for the example of ensembles where adaptive SEP protocol extract more classical information than local. Study of multi-copy state discrimination for higher dimensional and multipartite Hilbert spaces might reveal several other interesting features. One can consider more exotic adaptive protocols that allow to interact back and forth with the different copies by means of non-projective measurements. We believe that even under this more exotic protocol our two-copy local indistinguishable ensembles will remain indistinguishable. Formal proof of this assertion demands further research. 

\begin{acknowledgements}
MB acknowledges research grant of INSPIRE-faculty fellowship from the Department of Science and Technology, Government of India. MA acknowledges support from the CSIR
project 09/093(0170)/2016- EMR-I. SSB is supported by the National Natural Science Foundation of China through grant 11675136, the Foundational Questions Institute through grant FQXiRFP3-1325, the Hong Kong Research Grant Council through grant 17300918, and the ID\#61466 grant from the John Templeton Foundation, as part of the ``The Quantum Information Structure of Spacetime (QISS)" Project (\textcolor{blue}{qiss.fr}). The opinions expressed in this publication are those of the author(s) and do not necessarily reflect the views of the John Templeton Foundation.
\end{acknowledgements}

\onecolumngrid
\begin{center}
{\large Supplemental Material}
\end{center}
\section{Proof of Corollary 1}
\begin{proof}
Ensembles containing no entangled state are in category (i) of Theorem 1. There can not be any orthogonal basis in $(\mathbb{C}^2)^{\otimes2}$ that contain only one entangled state \cite{Bennett99(1)}. For the ensembles containing two entangled states the elimination round at the first copy level has to be done in a way so that one entangled states gets eliminated. Note that this can always be done (see {\it Theorem 1} of Ref. \cite{Bandyopadhyay11}) which consequently leaves these ensembles in category (ii) of Theorem 1.  
\end{proof}
\section{Analysis of the general case: $\mathcal{A}_{\gamma}^{[\alpha,\beta]}$}
Recall the set of four orthogonal states $\mathcal{A}_{\gamma}^{[\alpha,\beta]}$ as introduced in the main text. Among them the states $\ket{a_{1}}\equiv\ket{\phi^{-}(\alpha)}$ and $\ket{a_{2}}\equiv\ket{\psi^{-}(\beta)}$ are entangled $\forall~\alpha,\beta\in(0,\frac{\pi}{2})$. Consider the state $\ket{a_{3}}$ given by,
\begin{eqnarray}
\nonumber&\ket{a_{3}}&=S_\gamma\ket{\phi^{+}_\alpha}+C_\gamma \ket{\psi^{+}_\beta}\\\nonumber
&&=\sin\gamma\sin\alpha\ket{00}+\sin\gamma\cos\alpha\ket{11}+\cos\gamma\sin\beta\ket{01}+\cos\gamma\cos\beta\ket{10}\\\nonumber
&&=\ket{0}\otimes(\sin\gamma\sin\alpha\ket{0}+\cos\gamma\sin\beta\ket{1})+\ket{1}\otimes(\sin\gamma\cos\alpha\ket{1}+\cos\gamma\cos\beta\ket{0})\\\label{eq1}
&&=\mathcal{N}(\eta_{0})\ket{0\eta_{0}}+\mathcal{N}(\eta_{1})\ket{1\eta_{1}},
\end{eqnarray}
where, $\mathcal{N}(\eta_{0}):=\sqrt{\sin^{2}\gamma \sin^{2}\alpha + \cos^{2}\gamma \sin^{2}\beta},~\mathcal{N}(\eta_{1}):=\sqrt{\sin^{2}\gamma \cos^{2}\alpha + \cos^{2}\gamma \cos^{2}\beta}$ and $\ket{\eta_{0}},~\ket{\eta_{1}}$ are the normalized states. The state in Eq.\eqref{eq1} is separable, \textit{if and only if} $\ket{\eta_{0}}\equiv\ket{\eta_{1}}$ up to a global phase, which implies
\begin{eqnarray}
\nonumber\frac{\sin\gamma \sin\alpha}{\sqrt{\sin^{2}\gamma \sin^{2}\alpha + \cos^{2}\gamma \sin^{2}\beta}}&=&\frac{\cos\gamma \cos\beta}{\sqrt{\sin^{2}\gamma \cos^{2}\alpha + \cos^{2}\gamma \cos^{2}\beta}},\\\label{eq2}
\implies\tan^{2}\gamma&=&\frac{\sin 2\beta}{\sin 2\alpha}.
\end{eqnarray}
Similarly, separability of $\ket{a_{4}}=C_\gamma\ket{\phi^{+}_\alpha}-S_\gamma \ket{\psi^{+}_\beta}$ yields,
\begin{equation}\label{eq3}
\tan^{2}\gamma=\frac{\sin 2\alpha}{\sin 2\beta}.
\end{equation}
Notably, for $\gamma=\frac{\pi}{4}$ the conditions $\alpha=\beta$ or, $\alpha+\beta=\frac{\pi}{2}$ makes both of $\ket{a_{3}}$ and $\ket{a_{4}}$ separable. For generic $\gamma\in[0,\frac{\pi}{2}]$ and $\alpha, \beta\in (0,\frac{\pi}{2})$, separability of the states $\ket{a_3}~\&~\ket{a_4}$ is depicted in the Fig.\ref{fig1}.
\begin{figure}[h!]
\includegraphics[scale=0.4]{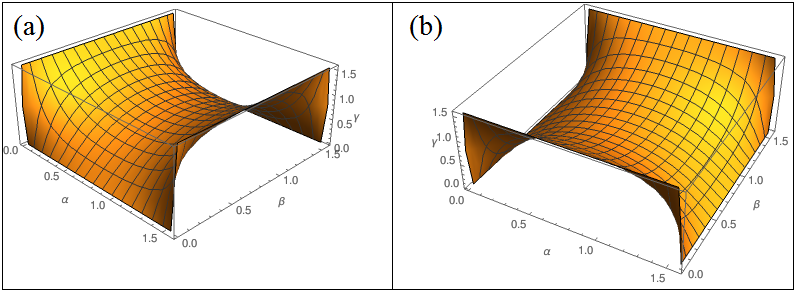}
\caption{(Color online) Parametric plot for $\{\alpha,\beta,\gamma\}$. (a) Every tuple of $(\alpha,\beta,\gamma)$ lying on the curved surface makes $\ket{a_{3}}$ separable. (b) Every tuple of $(\alpha,\beta,\gamma)$ lying on the curved surface makes $\ket{a_{4}}$ separable.}\label{fig1}
\end{figure}

Further in this section, we will calculate concurrence \cite{Wootters1998} of the states $\ket{a_{i}}\in\mathcal{A}_{\gamma}^{[\alpha,\beta]}$. Any two-qubit pure state $\ket{\psi}$ can be written as $\ket{\psi}=(\mathbf{I}\otimes\Psi)\ket{\phi^{+}}$ where $\ket{\phi^{+}}:=\frac{1}{\sqrt{2}}(\ket{00}+\ket{11})$ and $\Psi$ is a $2\times2$ matrix. The concurrence for $\ket{\psi}$ is then given by $C(\psi)=|\text{det}(\Psi)|$. Representing the states as $\ket{a_i}=(\mathbf{I}\otimes\mathbf{a_i})\ket{\phi^{+}}$ a simple numerical exercise yields,
\begin{equation}\label{eq5}
\begin{aligned}
&~~~~~~~~~\mathbf{a_1}:=
\begin{pmatrix}
\sqrt{2}\cos\alpha & 0\\
0 & -\sqrt{2}\sin\alpha
\end{pmatrix},~~~~~~~~~~~~~~~~~~~
\mathbf{a_2}:=
\begin{pmatrix}
0 & -\sqrt{2}\sin\beta \\
\sqrt{2}\cos\beta & 0
\end{pmatrix},\\\\
&\mathbf{a_3}:=
\begin{pmatrix}
\sqrt{2}\sin\gamma\sin\alpha & \sqrt{2}\cos\gamma\cos\beta\\
\sqrt{2}\cos\gamma\sin\beta & \sqrt{2}\sin\gamma\cos\alpha
\end{pmatrix},~~
\mathbf{a_4}:=
\begin{pmatrix}
\sqrt{2}\cos\gamma\sin\alpha & -\sqrt{2}\sin\gamma\sin\beta\\
-\sqrt{2}\sin\gamma\cos\beta & \sqrt{2}\cos\gamma\cos\alpha
\end{pmatrix}.
\end{aligned}
\end{equation}
Thus we have,
\begin{equation}\label{eq4}
\begin{aligned}
&C(a_{1})=|\det(\mathbf{a_1})|=\sin 2\alpha,\\
&C(a_{2})=|\det(\mathbf{a_2})|=\sin 2\beta,\\
&C(a_{3})=|\det(\mathbf{a_3})|=|\cos^{2}\gamma \sin 2\beta-\sin^{2}\gamma \sin 2\alpha|,\\
&C(a_{4})=|\det(\mathbf{a_4})|=|\cos^{2}\gamma \sin 2\alpha-\sin^{2}\gamma \sin 2\beta|.
\end{aligned}
\end{equation}
Evidently, for $\alpha,\beta\in(0,\frac{\pi}{2})$ the concurrences $C(a_{1})$ and $C(a_{2})$ are both positive. However, depending upon the triplet $(\alpha,\beta,\gamma)$ we can consider following four different regions for Eq.\eqref{eq4} (see Fig\ref{fig2} for better clarification):
\begin{itemize}
\item[] {\bf Region I} ($\mathcal{R}_{I}$):  $\frac{\sin 2\beta}{\sin 2\alpha}\geq\tan^{2}\gamma\geq\frac{\sin 2\alpha}{\sin 2\beta}$. Here, $C(a_{3})=\cos^{2}\gamma \sin 2\beta - \sin^{2}\gamma \sin 2\alpha$ and $C(a_{4})=\sin^{2}\gamma \sin 2\beta- \cos^{2}\gamma \sin 2\alpha$.
\item[] {\bf Region II} ($\mathcal{R}_{II}$): $\frac{\sin 2\alpha}{\sin 2\beta}\geq\tan^{2}\gamma\geq\frac{\sin 2\beta}{\sin 2\alpha}$. Here, $C(a_{3})=\sin^{2}\gamma \sin 2\alpha-\cos^{2}\gamma \sin 2\beta$ and $C(a_{4})=\cos^{2}\gamma \sin 2\alpha-\sin^{2}\gamma \sin 2\beta$.
\item[] {\bf Region III} ($\mathcal{R}_{III}$): $\tan^{2}\gamma\geq\max\left\{\frac{\sin 2\alpha}{\sin 2\beta},\frac{\sin 2\beta}{\sin 2\alpha}\right\}$. Here, $C(a_{3})=\sin^{2}\gamma \sin 2\alpha-\cos^{2}\gamma \sin 2\beta$ and $C(a_{4})=\sin^{2}\gamma \sin 2\beta-\cos^{2}\gamma \sin 2\alpha$. 
\item[] {\bf Region IV} ($\mathcal{R}_{IV}$): $\tan^{2}\gamma\leq\min\left\{\frac{\sin 2\alpha}{\sin 2\beta},\frac{\sin 2\beta}{\sin 2\alpha}\right\}$. Here, $C(a_{3})=\cos^{2}\gamma \sin 2\beta-\sin^{2}\gamma \sin 2\alpha$ and $C(a_{4})=\cos^{2}\gamma \sin 2\alpha-\sin^{2}\gamma \sin 2\beta$.
\end{itemize}
Notably, the saturation of either inequality implies the corresponding state is product, as mentioned earlier.
\begin{figure}[!h]
\includegraphics[scale=0.4]{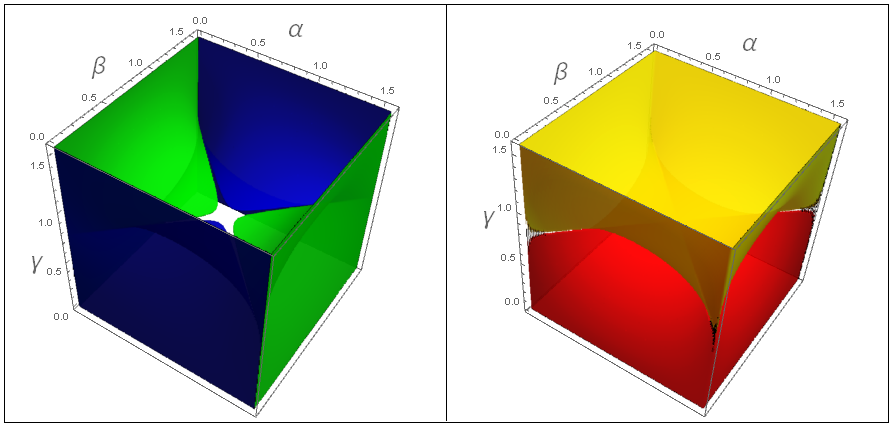}
\caption{(Color online) All the four regions depending upon $\{\alpha,\beta,\gamma\}$. \textcolor[rgb]{0,0.7,0}{Green} stands for the region $\mathcal{R}_{I}$, \textcolor[rgb]{0,0,1}{blue} for $\mathcal{R}_{II}$, \textcolor[rgb]{0.99, 0.64, 0.42}{yellow} for $\mathcal{R}_{III}$, and \textcolor[rgb]{1,0,0}{red} for $\mathcal{R}_{IV}$.}\label{fig2}
\end{figure}

\section{Proof of the Theorem(s) \& Proposition(s)}
\subsection{Proof of Proposition 1}
To prove the Proposition 1 we need to check NPT-ness of the projectors $\mathbb{P}_{ij}=\ket{a_{i}}\bra{a_{i}}+\ket{a_{j}}\bra{a_{j}}, \forall~i\neq j$. In the following we do this checking for general consideration, {\it i.e.} for the set $\mathcal{A}_{\gamma}^{[\alpha,\beta]}$. For instance, for arbitrary $\gamma$, eigenvalues for $\mathbb{P}_{12}^{\mathsf{T}_p}$ and $\mathbb{P}_{34}^{\mathsf{T}_p}$ turn out to be, 
\begin{subequations}\label{eq8}
\begin{align}
e_1=\frac{1}{4}\left(2-\sqrt{2}\left[2+\cos4\alpha-\cos4\beta\right]^{\frac{1}{2}}\right),~~~~
e_2=\frac{1}{4}\left(2+\sqrt{2}\left[2+\cos4\alpha-\cos4\beta\right]^{\frac{1}{2}}\right),\\
e_3=\frac{1}{4}\left(2-\sqrt{2}\left[2-\cos4\alpha+\cos4\beta\right]^{\frac{1}{2}}\right),~~~~
e_4=\frac{1}{4}\left(2+\sqrt{2}\left[2-\cos4\alpha+\cos4\beta\right]^{\frac{1}{2}}\right).
\end{align}
\end{subequations}
Here $\mathsf{T}_p$ denotes partial transposition. Clearly, $e_2~\&~e_4$ are always positive. However, almost everywhere in parameters space $(\alpha,\beta,\gamma)$ one of the eigenvalues either $e_1$ or $e_3$ is negative since $e_1e_3<0$ unless $\cos4\alpha=\cos4\beta$ (see Fig.\ref{fig3}) and hence ensures NPT-ness of the projectors $\mathbb{P}_{12}$ and $\mathbb{P}_{34}$.   
\begin{figure}[!htb]
\includegraphics[scale=0.5]{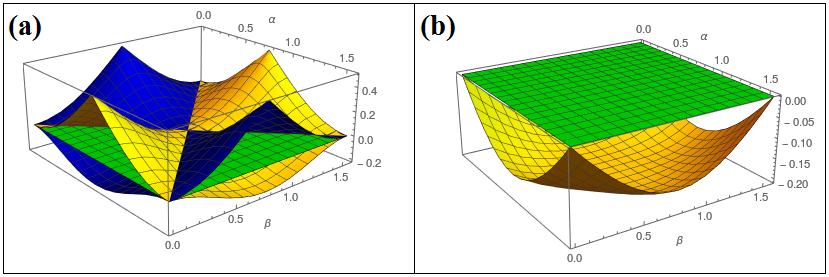}
\caption{(color online) \textcolor[rgb]{0,.5,0}{Green} corresponds to the zero ($z=0$) surface. (a) Eigenvalues $e_1$ (\textcolor{blue}{blue} surface) and $e_3$ (\textcolor[rgb]{0.99, 0.64, 0.42}{yellow} surface) of Eq.(\ref{eq8}).  Almost everywhere one of $e_1,e_3$ is negative. (b) Eigenvalue $e_3$ of Eq.(\ref{eq9}) (\textcolor[rgb]{0.99, 0.64, 0.42}{yellow} surface).}\label{fig3}
\end{figure}
For the other choices of the sub-spaces the generic expression for the eigenvalues of the partial transposition of the projectors onto them is much complicated. For instance, $\gamma=\pi/4$, eigenvalues of $\mathbb{P}_{13}^{\mathsf{T}_p},~\mathbb{P}_{14}^{\mathsf{T}_p},,~\mathbb{P}_{23}^{\mathsf{T}_p}$, and $\mathbb{P}_{24}^{\mathsf{T}_p}$ are given by
\footnotesize
\begin{subequations}\label{eq9}
\begin{align}
e_1&=\frac{1}{8} \left(4-\left(16-2 \sqrt{2} \left[(\sin (2 \alpha )+\sin (2 \beta ))^2 (18-12 \sin (2\alpha ) \sin (2\beta )-\cos (4 \alpha )-\cos (4 \beta ))\right]^{\frac{1}{2}}\right)^{\frac{1}{2}}\right),\\
e_2&=\frac{1}{8} \left(4+\left(16-2 \sqrt{2} \left[(\sin (2 \alpha )+\sin (2 \beta ))^2 (18-12 \sin (2\alpha ) \sin (2\beta ) -\cos (4 \alpha )-\cos (4 \beta ))\right]^{\frac{1}{2}}\right)^{\frac{1}{2}}\right),\\
e_3&=\frac{1}{8} \left(4-\left(16+2 \sqrt{2} \left[(\sin (2 \alpha )+\sin (2 \beta ))^2 (18-12 \sin (2\alpha ) \sin (2\beta ) -\cos (4 \alpha )-\cos (4 \beta ))\right]^{\frac{1}{2}}\right)^{\frac{1}{2}}\right),\\
e_4&=\frac{1}{8} \left(4+\left(16+2 \sqrt{2} \left[(\sin (2 \alpha )+\sin (2 \beta ))^2 (18-12 \sin (2\alpha ) \sin (2\beta ) -\cos (4 \alpha )-\cos (4 \beta ))\right]^{\frac{1}{2}}\right)^{\frac{1}{2}}\right).
\end{align}
\end{subequations}
\normalsize
Eigenvalue $e_3$ is negative for almost all values of $\alpha,\beta$ as depicted in Fig.\ref{fig3} ($e_3<0$ unless $\sin2\alpha=-\sin2\beta$). Furthermore, numerical verification (in MATHEMATICA) ensures NPT-ness of these projectors almost everywhere in the parameter space $(\alpha,\beta,\gamma)$.

\subsection{Proof of Theorem 3} 
\begin{proof}
An orthogonal basis $\mathcal{B}^{[3E1P]}$ containing three entangled states and one product state cannot be in category (i) of the Theorem 1. Furthermore $\mathcal{B}^{[3E1P]}$ cannot be in category (ii) of the Theorem 1, as any $1~vs~3$ grouping ends up with atleast two entangled states in the three-state group. Finally, $\mathcal{B}^{[3E1P]}$ cannot be in category (iii) of the Theorem 1 too. Any $2~vs~2$ grouping ends up with a group containing one product and one entangled state making the projector on this spanning subspace entangled (see {\bf Theorem  2} of Ref. \cite{Horodecki03(1)}). This completes the proof.  
\end{proof}
For the set $\mathcal{A}_{\gamma}^{[\alpha,\beta]}$, if we choose $\gamma=\tan^{-1}\left(\sqrt{\frac{\sin2\alpha}{\sin2\beta}}\right)$ or $\gamma=\tan^{-1}\left(\sqrt{\frac{\sin2\beta}{\sin2\alpha}}\right)$ and $\alpha,\beta\in(0,\pi/2)$, then the set contains exactly one product state and three entangled states. These constitute example of ensembles stated in Theorem 3.
\section{(In)distinguishability under adaptive SEP}
\subsection{Proof of Theorem 4}
\begin{proof}
An orthonormal basis of $(\mathbb{C}^2)^{\otimes2}$ is single-copy distinguishable under SEP {\it if and only if} all the basis states are product states \cite{Duan09} (see also Proposition 3 of Ref.\cite{Chitambar14}). For a basis containing entangled state(s), an adaptive discrimination protocol of SEP on the first copy will lead to one of the following two conclusions: (i) the state is from some group $\mathcal{G}_{ij}$ or from its complement group $\mathcal{G}^{\mathsf{C}}_{ij}$, (ii) it is from some group $\mathcal{G}_{l}$ or from its complement group $\mathcal{G}^{\mathsf{C}}_{l}$; $i,j,l\in\{1,2,3,4\}~\&~i\neq j$. For the first case the protocol will discriminate the corresponding density operator $\rho_{ij}$ and $\rho_{ij}^\perp$. The condition for perfect discrimination of two such rank two density operators under SEP is same as LOCC, {\it i.e.} they can be distinguished under SEP {\it if and only if} the projectors onto the supports of each of the density operators are separable \cite{Chitambar14}. Proposition 1, however, restrains this possibility. In the second case the protocol will sometime lead to the conclusion that the state is in the group $\mathcal{G}^{\mathsf{C}}_{l}$. In that case, perfect discrimination needs to be done from the second copy. An ensemble $\{\ket{\psi_i}\}_{i=1}^3$ of three orthogonal states of $(\mathbb{C}^2)^{\otimes2}$ is perfectly distinguishable by separable operations {\it if and only if}: (a) $\psi_k\phi^{-1}$ has two anti-parallel eigenvalues for each entangled state $\psi_i$, and (b) $\sum_{i=1}^3 C(\psi_i)= C(\phi)$ \cite{Duan09}. Here $\ket{\phi}$ is the state orthogonal to the subspace spanned by $\{\ket{\psi_i}\}_{i=1}^3$, and $C(\psi)=|\det(\psi)|$ denotes the concurrence of the state $\ket{\psi}:=(\mathbb{I}\otimes\psi)\ket{\phi^+_{\frac{\pi}{4}}}$, with $\psi$ being a $2\times 2$ matrix. In the following we check these conditions. 

To complete the argument, let us consider that after the protocol on first copy, none of the states is eliminated, rather the knowledge regarding the ensemble gets updated, {\it i.e.} the uniform ensemble  $\mathcal{E}_4\equiv\{1/4,\ket{\psi_i}\}$ ends up into some non-uniform ensemble $\mathcal{E}^\prime_4\equiv\{p_i,\ket{\psi_i}~|~p_i\neq0~\forall~i\}$. In such a case, the ensemble $\mathcal{E}^\prime$ must be discriminated from the second copy of the state. However this cannot be done as the argument presented above is independent of the prior probability distributions. This completes the proof.
\end{proof}
We consider the generic ensemble $\mathcal{A}_{\gamma}^{[\alpha,\beta]}$ and analyze its distinguishability under adaptive SEP. 

{\bf Region I:} Straightforward calculation yields $\mathbf{a_k}\mathbf{a_2}^{-1}$ has anti-parallel eigenvalues for all $k\in\{1,3,4\}$. Furthermore, 
\begin{equation*}
C(a_{1})+C(a_{3})+C(a_{4}) = \sin 2\beta=C(a_{2}).
\end{equation*}
Thus in the adaptive scheme, from the first copy it is determined whether the given state is $\ket{a_{2}}$ or not. In fact this can be done by LOCC \cite{Bandyopadhyay11}. From the second copy the rest three state can be perfectly discriminated under SEP \cite{Duan09}.

{\bf Region II:} This case is same as {\bf Region I} with $\ket{a_{2}}$ interchanged with $\ket{a_{1}}$. 

{\bf Region III:} Note that, 
\begin{subequations}
\begin{align}
C(a_{1})=\sum_{i\neq 1}C(a_{i})&\implies\tan^2\gamma=\frac{\sin 2\alpha}{\sin 2\beta},\\
C(a_{2})=\sum_{i\neq 2}C(a_{i})&\implies\tan^2\gamma=\frac{\sin 2\beta}{\sin 2\alpha},\\
C(a_{3})=\sum_{i\neq 3}C(a_{i})&\implies~~~~~~\beta=0,\\
C(a_{4})=\sum_{i\neq 4}C(a_{i})&\implies~~~~~~\alpha=0.
\end{align}
\end{subequations}
None of the above conditions can be satisfied in $\mathcal{R}_{III}$ with all of states in $\mathcal{A}_{\gamma}^{[\alpha,\beta]}$ entangled. Hence, none of the possible $1:3$ partition for the set $\mathcal{A}_{\gamma}^{[\alpha,\beta]}$ can be discriminated perfectly under adaptive SEP and therefore under adaptive LOCC. 
 
{\bf Region IV:}
\begin{subequations}
\begin{align}
C(a_{1})=\sum_{i\neq 1}C(a_{i})&\implies\tan^2\gamma=\frac{\sin 2\beta}{\sin 2\alpha},\\
C(a_{2})=\sum_{i\neq 2}C(a_{i})&\implies\tan^2\gamma=\frac{\sin 2\alpha}{\sin 2\beta},\\
C(a_{3})=\sum_{i\neq 3}C(a_{i})&\implies~~~~~~\alpha=0,\\
C(a_{4})=\sum_{i\neq 4}C(a_{i})&\implies~~~~~~\beta=0.
\end{align}
\end{subequations}
These imply impossibility of discriminating the set of the entangled states ensemble $\mathcal{A}_{\gamma}^{[\alpha,\beta]}$ in $\mathcal{R}_{IV}$ under adaptive SEP. 
\begin{figure}[!htb]
\includegraphics[scale=0.4]{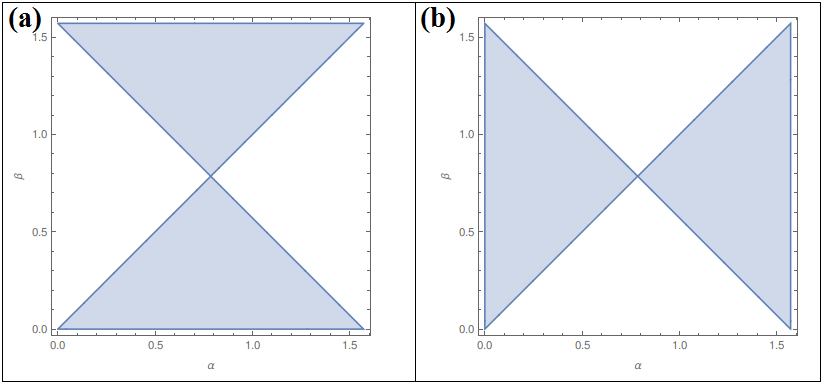}
\caption{(color online) (a) The \textcolor{blue}{blue} corresponds to $\mathcal{R}_{III}(\alpha,\beta)\subset\mathcal{R}_{III}$ where the ensemble $\mathcal{A}_{\gamma^\star}^{[\alpha,\beta]}$ can be perfectly distinguished under adaptive SEP. (b) Similar region $\mathcal{R}_{IV}(\alpha,\beta)\subset\mathcal{R}_{IV}$.}\label{fig4}
\end{figure}

The parametric class $\mathcal{A}_{\gamma}^{[\alpha,\beta]}$ also possess some interesting consequences. Note that, in $\mathcal{R}_{I}$ and $\mathcal{R}_{II}$ no adaptive LOCC scheme can accomplish the discrimination task perfectly while adaptive SEP can do so. For $\gamma=\gamma^\star=\tan^{-1}\left(\sqrt{\frac{\sin 2\alpha}{\sin 2\beta}}\right)$, the state $\ket{a_{4}}$ is a product state. Now for some values of $(\alpha,\beta)$ there will be some region in $\mathcal{R}_{III}$, let denote this region by $\mathcal{R}_{III}(\alpha,\beta)$, where the set $\mathcal{A}_{\gamma^\star}^{[\alpha,\beta]}$ can be discriminated under adaptive SEP. If can be checked that $C(a_{1})=\sum_{i\neq 1}C(a_{i})$ and also $\forall k \in \{2,3\},~\mathbf{a_{k}}\mathbf{a_{1}}^{-1}$ has two anti-parallel eigenvalues in $\mathcal{R}_{III}$. For same $\gamma^\star$ in the region $\mathcal{R}^{\mathsf{C}}_{III}(\alpha,\beta)$ complimentary to $\mathcal{R}_{III}(\alpha,\beta)$ the set $\mathcal{A}_{\gamma^\star}^{[\alpha,\beta]}$ can not be discriminated under adaptive SEP (see Fig. \ref{fig4}). A similar analysis can be done in $\mathcal{R}_{IV}$ and also for $\gamma=\tan^{-1}\left(\sqrt{\frac{\sin 2\beta}{\sin 2\alpha}}\right)$.


\begin{thebibliography}{99}
\bibitem{Self} Collaboration among certain groups of the four parties can reveal at-most $1$-bit data which is only half of the desired message. It is also not hard to argue that collaboration among any five parties cannot reveal $2$-bit message exactly. It should be noted that security of the protocol is assured under the adaptive protocol considered in Definition \ref{def1}. A more generic protocol might further demand limitation on the LOCC operations among different groups of the parties.
\end{thebibliography}
 \end{document}